\documentclass[a4paper]{article}

\usepackage{amsmath}
\usepackage{amsfonts}
\usepackage{amssymb}
\usepackage{graphicx}
\usepackage[all]{xy}

\newtheorem{theorem}{Theorem}

\newtheorem{lemma}[theorem]{Lemma}

\newtheorem{proposition}[theorem]{Proposition}

\newenvironment{proof}[1][Proof]{\textbf{#1.} }{\ \rule{0.5em}{0.5em}}

\begin{document}

\title{Heterogeneous Beliefs Model of Stock Market Predictability}
\author{Jiho Park\thanks{Thanks to Dimitri Vayanos, Christian Julliard, Kathy Yuan, Igor Makarov, and Andrea Tamoni for their insights and guidance.}
\\London School of Economics}
\date{\today}
\maketitle

\begin{abstract}
This paper proposes a theory of stock market predictability patterns based on a model of heterogeneous beliefs. In a discrete finite time framework, some agents receive news about an asset's fundamental value through a noisy signal. The investors are heterogeneous in that they have different beliefs about the news signal. A momentum in the stock price arises from those agents who incorrectly underestimate the signal accuracy, dampening the initial price impact of the signal. A reversal in price occurs because the price reverts to the fundamental value in the long run. An extension of the model to multiple assets case predicts co-movement and lead-lag effect, in addition to cross-sectional momentum and reversal. The heterogeneous beliefs of investors about news demonstrate how the main predictability anomalies arise endogenously in a model of bounded rationality.
\end{abstract}

\section{Introduction}

Many empirical works over the past decades have documented several stock market patterns that challenge the efficient market foundation built on the CAPM. The two prominent ones among these anomalies are `momentum' and `reversal' effects. The momentum pattern predicts that stocks that have performed well in the past continue to do so, and those that did poorly continue to show poor performance in the future. Starting with Jegadeesh and Titman(1993), many subsequent papers have shown this empirical pattern of positive autocorrelation in prices in the medium term future. The reversal effect occurs when over a longer horizon, the reverse happens. There is a negative autocorrelation in stock prices in the long run as documented by De Bondt and Thaler(1985) and other works. These predictable patterns present unexplained arbitrage opportunities demonstrated by the significant excess returns of the portfolios based on the momentum and the reversal strategies.\\
\indent In explaining momentum and reversal, theoretical research has followed three main approaches. The first and the initial effort to explain the anomalies in terms of risk compensation have not been successful. Griffin, Ji, and Martin(2003) demonstrated this by showing that the excess returns of the momentum and reversal portfolios are not correlated with any of the risk factors. Another group of theoretical works has tried to explain the phenomena through behavioral finance. The main works include Daniel, Hirshleifer, and Subrahmanyam(1998) and Barberis, Shleifer, and Vishny(1998) which have attributed the anomalies to the underreaction and overreaction due to investor sentiments. Empirical studies show that behavioral explanation could partially explain momentum and reversal but is not always convincing as Jegadeesh and Titman(2001) argued. The last and the recent theoretical approach employs models with bounded rationality, which this paper focuses on.\\
\indent A fully rational model approach has not yet been successful in explaining both momentum and reversal in a unified way. On the one hand, Berk, Green, and Naik(1999) and Johnson(2002) present rational models resulting in momentum but not reversal. On the other hand, the rational model in Lewellen and Shanken(2002) explains reversal but not momentum, and Cespa and Vives(2012) model can explain the one or the other but not both. Despite these shortcomings, Vayanos and Woolley(2012) show a rational model of delegated portfolio management that produce both momentum and reversal, demonstrating that rational models can be potentially successful in explaining momentum and reversal in a unified manner.\\
\indent It is especially difficult to generate momentum using a fully rational equilibrium model as demonstrated by Makarov and Rytchkov(2012). Their paper shows that even with the presence of frictions such as information asymmetry, the price autocorrelation is equal to the expected value of the current supply times the current price change. Since supply and price change are inversely related, the price autocorrelation is in most cases negative. Even in a more general setting, the paper shows the robustness of this result.\\
\indent The intuition of this paper's model is as follows. In a four period model, a proportion of the investors receive a noisy signal about a risky asset's fundamental value after the first period. The heterogeneity exists in that the informed and the uninformed agents have different beliefs about the noisy supply, which matters for the uninformed in inferring the signal from price. The informed believe that the supply is noisier than the true distribution known by the uninformed. Suppose the informed had received a relatively good signal. Then, the price will increase but under-react to the news as the informed believe that the supply is noisier and expect that the uninformed will react less to the price change.\\
\indent In the next penultimate period, the price increases further as the informed investors no longer care about the noisy supply in submitting their demand. They required the distribution of the noisy supply in forecasting the uninformed demand and now uses only their signal to predict the final period value. The continued price change pushes the price above the final payoff, creating an overreaction. When the actual payoff is realized in the final period, the price converges, resulting in the price reversal. While the stochastic supply is enough to generate reversal in a symmetric information case, the asymmetric information exacerbates the reversal effect. The heterogeneous beliefs about news in the form of a noisy signal creates both momentum and reversal in a bounded rationality model.\\
\indent The application of information asymmetry to asset pricing has been developed by several papers including Wang(1993). While his information asymmetry model is used to explain several other asset pricing phenomena, this paper instead seeks to explain both momentum and reversal in a unified model. A recent application of information asymmetry model to liquidity is Vayanos and Wang(2012), from which this paper borrows the basic setup. Hong and Stein(1999) is the closest to this paper in its explanation of stock market predictability through the interaction of the news-watchers and the momentum traders. Their setup, however, is different from the information asymmetry setting of this paper. A direct application of information asymmetry to momentum and reversal appears in Albuquerque and Miao(2014), but in contrast to this paper, their model employs a different model with the uninformed acting as contrarians.\\
\indent The contributions of this paper is in two-fold. First, this paper employs the simplest asymmetric information framework to explain the stock market patterns. Without making further assumptions about investor's behavior or trading, this paper shows how the natural setting of information asymmetry is sufficient in explaining reversal coupled with heterogeneous beliefs. Second, the model seeks to explain price autocorrelation in a context where there is no exogenous time-varying mechanism driving the price change in a certain direction. The same heterogeneity persists throughout the model but creates price dynamics over different periods. Also, the basic framework could be easily extended to explain other stock market patterns such as co-movement and lead-lag effects. These contributions together demonstrate how strong information asymmetry could be in explaining the stock price trends in a unified framework.\\
\indent This paper is organized into five sections. Section 2 outlines the single asset model to provide a clear intuition for momentum and reversal patterns. The third section extends the model to multiple assets to analyze the cross-sectional momentum and reversal, illustrating the several of the related empirical patterns as well. The penultimate section uses the multiple asset model to explain other asset pricing patterns, and the final section concludes.

\section{Single Asset Model}

\subsection{Basic Framework}
In this finite discrete time model, there are four periods denoted by $t \in \{0,1,2,3\}$ with a single risky asset and a risk-less asset. The risky asset pays off $D\sim\mathcal{N}(\bar{D},\sigma^2_D)$ in the final period $t=3$ with no
interim payments. The supply of the risky asset is stochastic $1+\theta_t$ with $\theta_t\sim\mathcal{N}(0, \sigma^2_{\theta,i})$ in $t=1,2$ and fixed at 1 in the ex-ante and final periods $t=0,3$. The risk-less asset has a certain payoff of 1 in the last period and is in fixed supply of $B$. With no time discounting and using the risk-less asset as the numeraire, the risky asset's price is $S_t$ in period $t$. Initially, each agent is endowed with per capita amount of the risky and the risk-less assets.\\
\indent The economy has a unit measure of agents of two types $i\in\{1,2\}$ with $\pi$ proportion of the $i=1$ agents, and $1-\pi$ of $i=2$. The informed type $i=1$ receive a news in the form of a homogeneous noisy signal $s\sim\mathcal{N}(D,\sigma^2_s)$ about the risky asset payoff between periods 0 and 1. The uninformed $i=2$ can only infer the signal from prices. The information asymmetry is coupled with incomplete markets in the model because the supply of the risky asset is stochastic. Hence, the uninformed investors cannot perfectly infer the signal from the prices due to the additional uncertainty. This incomplete market setting is required for the information asymmetry to persist in the economy as proved by Grossman(1981).\\
\indent In this asymmetric information framework, the two types agree on all the hyper-parameters in the model except for the variance of the stochastic supply $\sigma^2_{\theta,i}$. Assume that $\sigma^2_{\theta,1}>\sigma^2_{\theta,2}$ so that the informed investors incorrectly and heterogeneously believe in a higher supply variance than the uninformed. The uninformed's distribution is the true one throughout the model, so $\sigma^2_\theta=\sigma^2_{\theta,2}$ with the two expressions used interchangeably throughout for convenience. The setting models the different beliefs about the signal through heterogeneous beliefs in the variance of the stochastic supply. Since the uninformed can only infer the signal with less accuracy due to the noisy supply, the heterogeneous beliefs about the noisy supply also applies to the distribution of the signal.\\
\indent In the ex-ante $t=0$, all agents face two uncertainties for the next period in the noisy signal $s$ and the stochastic supply $1+\theta_1$. Before period 1, news arrives through $s$, and agents trade based on their heterogeneous beliefs about the signal validity in $t=1$. This heterogeneous belief is the only friction in this asymmetric information economy as it is assumed that both types of agents are dogmatic about their beliefs. In period 1, the investors face an uncertainty in the noisy supply for the next period. In the last period which models the long run, the final payoff is realized, so $S_3=D$. While the information set for the informed $i=1$ is $\mathcal{F}_{t,1}=\{s,S_\tau,\sigma^2_{\theta,1}:\tau=0,\dots,t\}$, for the uninformed $i=2$, $\mathcal{F}_{t,2}=\{S_\tau,\sigma^2_\theta:\tau=0,\dots,t\}$ since knowing the price does not reveal both $s$ and $\theta_1$.\\
\indent Throughout this paper, subscripts for variables denote time $t$ and agent type $i$ in that order with any of the subscripts suppressed when clear from context. The agents are myopic and seek to maximize their wealth in the next period, which is assumed to simplify the computation and is without loss of generality as proven in a verification theorem in the appendix. All investors are assumed to know ex-ante their own types and the market structure, and to maximize an exponential utility function of $U(W) = -\exp\{-\alpha W\}$ with the risk aversion coefficient $\alpha$.

\subsection{Model Equilibrium}
In the following equilibrium, each investor chooses his risky asset demand in each period to maximize his wealth in the next period. The optimization is solved backwards starting with $t=2$ with the expectations and variances conditional on the information sets in the period. The maximization problem in time $t$ is as follows.
\begin{equation*}
\max_{x_t}E_t(U(W_{t+1}))=E_t(-\exp\{-\alpha(W_t+x_t(S_{t+1}-S_t))\})
\end{equation*}
where $x_t$ is the risky asset demand, and $W_t$, the wealth in period $t$. With the normally distributed signal, the first order condition gives the solution as
\begin{equation*}
x_t=\frac{E_t(S_{t+1})-S_t}{\alpha Var_t(S_{t+1})}
\end{equation*}
Following the method in Vayanos and Wang(2012), conjecture a price function that is affine in the signal $s$.
\begin{equation*}
S_t=a_t+b_t(s-\bar{D}+c_t\theta_t)
\end{equation*}
for $t=1,2$.\\
\indent With signal $s$, the normal distribution results in
\begin{align*}
&E_{2,1}(S_3)=E_1(D)=\bar{D}+\beta_s(s-\bar{D})\\
&Var_{2,1}(S_3)=Var_1(D)=\beta_s\sigma^2_s\\
&E_{2,2}(S_3)=E_2(D)=\bar{D}+\beta_\xi\xi_2\\
&Var_{2,2}(S_3)=Var_2(D)=\beta_\xi\sigma^2_\xi
\end{align*}
where $\beta_s=\sigma^2_D/(\sigma^2_D+\sigma^2_s)$, $\beta_\xi=\sigma^2_D/(\sigma^2_D+\sigma^2_s+c_2^2\sigma^2_\theta)$, and $\sigma^2_\xi=\sigma^2_s+c_2^2\sigma^2_\theta$.\\
\indent The market-clearing conditions are
\begin{align*}
x_{t,1}+x_{t,2}=1
\end{align*}
for $t=0,1,2$. Substituting in the necessary variables, the resulting equation is affine with respect to $s$ and $\theta_t$. The coefficients and the constant term need to equal zero for the equation to hold for any values of $s$ and $\theta_t$. The result is a system of equations for each coefficient and unknowns. The variables with asterisks denote those under the incorrect expectation of the informed; for example $a_2^*=E_1(a_2)$. $Var_i(S_1)$ denote the variance of $S_1$ under the expectation of type $i$ agent. The resulting prices for each period is found as follows.
\begin{proposition}
The prices are
\begin{align*}
S_3=&D\\
S_2=&a_2+b_2(s-\bar{D}+c_2\theta_2)\\
S_1=&a_1+b_1(s-\bar{D}+c_1\theta_1)\\
S_0=&\frac{\pi Var_2(S_1)a_1^*+(1-\pi)Var_1(S_1)a_1-\alpha Var_1(S_1)Var_2(S_1)}{\pi Var_2(S_1)+(1-\pi)Var_1(S_1)}
\end{align*}
where
\begin{align*}
a_2=&\bar{D}-\alpha\frac{(\beta_s-b_2)\beta_\xi(\sigma^2_s+c_2^2\sigma^2_\theta)+(\beta_\xi-b_2)\beta_s\sigma^2_s}{\beta_s-\beta_\xi}\\
b_2=&\frac{\pi\beta_\xi(\sigma^2_s+c_2^2\sigma^2_\theta)\beta_s+(1-\pi)\beta_s\sigma^2_s\beta_\xi}{\pi\beta_\xi(\sigma^2_s+c_2^2\sigma^2_\theta)+(1-\pi)\beta_s\sigma^2_s} &&c_2=\frac{\alpha\sigma^2_s}{\pi}\\
a_1=&\frac{(b_2-b_1)(a_2^*-\alpha (b_2^*)^2c_2^2\sigma_{\theta,1}^2)+(b_2^*-b_1)(a_2-\alpha b_2^2c_2^2\sigma^2_\theta)}{(b_2-b_1)+(b_2^*-b_1)}\\
b_1=&\frac{\pi b_2^2\sigma^2_\theta b_2^*+(1-\pi)(b_2^*)^2\sigma_{\theta,1}^2 b_2}{\pi b_2^2\sigma^2_\theta+(1-\pi)(b_2^*)^2\sigma_{\theta,1}^2} &&c_1=\frac{\alpha b_2^*c_2^2\sigma_{\theta,1}^2}{\pi}
\end{align*}
\end{proposition}
The proofs of all propositions and theorems are in the appendix.\\
\indent In the previous proposition, the constant terms $a_1$, $a_2$, and $S_0$ measure the ex-ante payoff expectation $\bar{D}$ and the risk-adjustments. The coefficients $b_1$ and $b_2$ represent the impact of the noisy signal $s$ on the prices. $b_1<b_2$, so the signal impact on price increases from $t=1$ to $t=2$. The initial signal impact is dampened by the incorrect belief of the informed investors about the noisy supply. As it turns out next, the inequality $b_1<b_2$ generates the momentum effect.

\subsection{Momentum and Reversal in Price}
With the prices for $t=0,1,2,3$, the equilibrium for the model has been computed. The following section explores the empirical implications of the model including the main result of this paper on momentum and reversal. The model proposed in this paper predicts momentum and reversal in the price level of a single risky asset. Hence, the stock price patterns in this paper are not the relative ones in the cross-section, but rather a time-series version of momentum and reversal for a single stock. Asness, Moskowitz, and Pedersen(2013) documents the anomalies in the time-series, so their work is the most relevant for the empirical connections in this paper. Also, it is worth noting that most empirical studies examine momentum and reversal in terms of returns, not price levels as this paper does. While required for a stronger empirical connection, this paper does not make this distinction with the main autocorrelation results still valid since returns and price levels are correlated.\\
\indent With the prices for each period, the autocorrelations of the prices can be calculated to show whether the model predicts momentum or reversal. The following measures of covariance are used to check if the model results are consistent with the observed momentum and reversal patterns.
\begin{align*}
&\gamma_m=Cov(S_2-S_1,S_1-S_0)\\
&\gamma_r=Cov(D-S_2,S_2-S_1)
\end{align*}
The empirical evidence requires that $\gamma_m>0$ and $\gamma_r<0$. The following proposition shows that the single asset model can result in a momentum in pricing.
\begin{proposition}
Momentum measure is calculated as follows.
\begin{align*}
\gamma_m=Cov(S_2-S_1,S_1-S_0)=(b_2-b_1)b_1\sigma^2_s-b_1^2c_1^2\sigma^2_\theta
\end{align*}
The measure is positive if
\begin{align*}
\pi b_2^2\sigma^2_s(b_2-b_2^*)-b_1c_1^2>0
\end{align*}
in which case, momentum in prices arises.
\end{proposition}
The momentum result depends on the condition in the above proposition. There are many parameters to pin down, so a precise result is not obtained. But the above proposition shows that the price momentum could arise in a asymmetric information setting with heterogeneous beliefs.\\
\indent The price autocorrelation depends especially on the coefficients $b_1$ and $b_2$. A rough outline of the idea is that $b_2$ is the weighted average of $\beta_s$ and $\beta_\xi$ and that $b_1$ is that of $b_2$ and $b_2^*$. For momentum, $b_2-b_1>0$ for the signal impact to be higher in $t=2$. This inequality holds if $b_2^*<b_2$, that is, the informed's incorrect expectation of $b_2$ is less than the true expectation. With $\sigma_{\theta,1}>\sigma_{\theta,2}$, the inequality is satisfied and momentum may occur under certain conditions.\\
\indent The momentum effect results when the heterogeneity in beliefs is large enough compared to the supply shock. In the above final expression, there are two opposing effects from the two frictions. On the one hand, the heterogeneous beliefs drives the momentum effect which is the first term in the autocorrelation measure. On the other hand, the stochastic supply creates the reversal effect reflected in the second term. As long as the heterogeneity friction is large compared to that of the supply, the first effect dominates the second, resulting in the momentum pattern.\\
\indent The momentum effect is larger if there are more informed traders in the economy as $\partial\gamma_m/\partial\pi>0$. As there are more agents with incorrect beliefs about the noisy supply distribution, the signal impact on price is further dampened in the first period, resulting in an increased momentum. In a similar manner, defining $\Delta\sigma^2_\theta=\sigma^2_{\theta,1}-\sigma^2_{\theta,2}$, the momentum in price is stronger with more heterogeneity in beliefs as $\partial\gamma_m/\partial\Delta\sigma^2_\theta>0$. As is usually the case in frictions, a higher risk-aversion $\alpha$ leads to a higher effect of the friction in momentum.\\
\indent The next proposition proves that reversal in price results from the model.
\begin{proposition}
Reversal measure is calculated as follows.
\begin{align*}
\gamma_r=Cov(D-S_2,S_2-S_1)=-b_2(b_2-b_1)\sigma^2_s-b_2^2c_2^2\sigma^2_\theta
\end{align*}
The measure is always negative as $b_2>b_2^*$.
\end{proposition}
\indent The reversal in price occurs because of the stationarity induced by the stochastic supply. In $t=1$, the stochastic supply $1+\theta_1$ is realized, which is expected to return to 1 in $t=2$ as $E_1(\theta_2)=0$. Likewise, the realized supply $1+\theta_2$ in $t=2$ is certain to revert back to 1 as the supply is fixed in $t=3$. Hence, the reversal effect is larger if the supply is noisier, $\partial\gamma_r/\partial\sigma^2_\theta$.\\
\indent A symmetric information case with noisy supply is sufficient to produce reversal in price, but the asymmetric information setting in this model amplifies the effect. The noisy signal is also stationary over the entire horizon as $s\to D$ in the last period $t=3$. Therefore, the additional stationarity in $s$ also contributes to the price reversal. So, noisier signal increases the effect of reversal as in $\partial\gamma_r/\partial\sigma^2_s$. As price is pulled away from the fundamental value due to the forces driving the initial momentum effect, the price reverts in the last period as the final payoff is realized as $D$.\\
\indent The single asset model provides the answer to the key question that make rational models difficult to produce momentum. Suppose the signal is positive, then why would rational investors want to buy less when the expected return has increased? In period 1, the informed agents who incorrectly believe that the expected return will not increase as much will submit smaller demand for the asset as the current price is higher than that justified by their dogmatic belief. The heterogeneous beliefs setup allows this seemingly irrational trading in a bounded rationality model.

\section{Multiple Asset Model}
\subsection{Framework Extension}
Extending the previous single asset model to a one with multiple assets is straightforward. The results in this section are simply the counterparts of the previous single asset case in vector forms. With the same setup and assumptions as before, now the economy has $N+1$ number of assets denoted by $m=\{0,1,\dots,N\}$. Asset $m=0$ is the previous risk-free asset, while the other assets $m$ are risky with final payoff $D_m\sim\mathcal{N}(\bar{D_m},\sigma^2_m)$. For each asset, the informed investors receive a homogeneous signal about the final payoff $D_m$ denoted as $s_m\sim\mathcal{N}(D,\sigma^2_{m,s})$. The $N\times N$ covariance matrices between the payoffs and the signals of the risky assets are $\Sigma_D$ and $\Sigma_s$ with row $m$ and column $n$ element $\sigma_{m,n}$, the covariance the variables between the assets $m$ and $n$. The subscripts denote time $t$, asset $m$, and agent type $i$ in that order, so the demand for a risky asset $m$ in period $t$ by type $i$ is denoted as $x_{t,m,i}$. The uppercase variables or those without the asset subscript will denote a vector of variables so that $S_t$ is the $N\times 1$ vector of prices at time $t$.\\
\indent Again, there is the $N\times 1$ vector of noisy supplies $\textbf{1}+\Theta_t\sim\mathcal{N}(\textbf{0},\Sigma_\theta)$ for $t=1,2$ where bold character denotes vectors such that $\textbf{0}$ is a $N\times1$ vector of 0s. The heterogeneous beliefs arises because the informed investors $i=1$ have incorrect distribution of the stochastic supply with covariance matrix $\Sigma_{\theta,1}$. The true one known by the uninformed $i=2$ is denoted as $\Sigma_{\theta,2}=\Sigma_\theta$ with the assumption that $\Sigma_{\theta,1}-\Sigma_{\theta,2}$ is positive definite. This friction creates the price momentum and reversal as before.

\subsection{Model Equilibrium}
The equilibrium is found using a similar method from the previous single asset case. Given the same exponential utility function, the optimization is solved backwards from $t=2$ using the same affine price function conjecture as before.
\begin{equation*}
S_t=A_t+B_t(s-\bar{D}+C_t\Theta_t)
\end{equation*}
for $t=1,2$ where the equation is in vector notation with a $N\times 1$ constant vector $A_t$, $N\times 1$ signal vector $s$, and a $N\times N$coefficient matrices $B_t$ and $C_t$.\\
\indent The expectations and variances for each belief are the same as before, only that they are now in vector notation such that $\beta_s$ is an $N\times N$ vector of $\beta_{m,s}$ for each asset $m$. The market-clearing conditions also have the same form, and substituting and setting the two coefficients equal to zero as before give the prices in the following proposition. $\Sigma_i(S_1)$ denote the covariance matrix of $S_1$ under type $i$ expectation.
\begin{proposition}
The prices are
\begin{align*}
S_2=&A_2+B_2(s-\bar{D}+C_2\Theta_2)\\
S_1=&A_1+B_1(s-\bar{D}+C_1\Theta_1)\\
S_0=&(\pi Var_2(S_1)+(1-\pi)Var_1(S_1))^{-1}(\pi Var_2(S_1)A_1^*\\
&+(1-\pi)Var_1(S_1)A_1-\alpha Var_1(S_1)Var_2(S_1))
\end{align*}
where
\begin{align*}
A_2=&\bar{D}-\alpha(\beta_s-\beta_\xi)^{-1}((\beta_s-B_2)(\Sigma_s+C_2\Sigma_\theta C_2')\beta_\xi'+(\beta_\xi-B_2)\Sigma_s\beta_s')\\
B_2=&(\pi\beta_\xi(\Sigma_s+C_2\Sigma_\theta C_2')+(1-\pi)\beta_s\Sigma_s)^{-1}\\
&\times(\pi\beta_\xi(\Sigma_s+C_2\Sigma_\theta C_2')\beta_s+(1-\pi)\beta_s\Sigma_s\beta_\xi)\\
C_2=&\alpha\pi^{-1}\Sigma_s\\
A_1=&((B_2-B_1)(A_2^*-\alpha B_2^*C_2\Sigma_{\theta,1}(B_2^*C_2)')\\
&+(B_2^*-B_1)(A_2-\alpha B_2C_2\Sigma_\theta(B_2C_2)')((B_2-B_1)+(B_2^*-B_1))^{-1}\\
B_1=&(\pi B_2\Sigma_\theta B_2'+(1-\pi)B_2^*\Sigma_{\theta,1}(B_2^*)')^{-1}\\
&\times(\pi B_2^*B_2\Sigma_\theta B_2'+(1-\pi)B_2B_2^*\Sigma_{\theta,1}(B_2^*)')\\
C_1=&\alpha\pi^{-1}B_2^*C_2C_2^*\Sigma_{\theta,1}
\end{align*}
\end{proposition}
The result is simply the one from the single asset case in vector notation. Hence, the economic meaning behind $A_t$, $B_t$, and $C_t$ are the same as in the previous case for $t=1,2$.

\subsection{Cross-sectional Momentum and Reversal}
In the single asset model, the momentum and reversal in price was in time-series. With the multiple asset extension, the autocorrelation of price changes can be illustrated cross-sectionally. This is more relevant for empirical connections as most of the empirical studies focus on relative momentum and reversal in the cross-section.
\begin{proposition}(Cross-sectional Momentum and Reversal)\\
The momentum measure is the matrix $\Gamma_m=Cov(S_2-S_1,S_1-S_0)$.
\begin{align*}
\Gamma_m=Cov(S_2-S_1,S_1-S_0)=(B_2-B_1)\Sigma_sB_1'-B_1C_1\Sigma_\theta C_1'B_1'
\end{align*}
If the matrix is positive semi-definite, a cross-sectional momentum results.\\
The reversal measure is the matrix $\Gamma_r=Cov(D-S_2,S_2-S_1)$.
\begin{align*}
\Gamma_r=Cov(D-S_2,S_2-S_1)=-(B_2-B_1)\Sigma_sB_2'-B_2C_2\Sigma_\theta(B_2C_2)'
\end{align*}
which is always negative semi-definite, resulting in a cross-sectional price reversal.
\end{proposition}
Again, the momentum and reversal results are the vector form counterparts of the previous single asset case. Similar reasoning gives the comparative statics of increased momentum with larger $\pi$ and $\Delta\Sigma_\theta=\Sigma_{\theta,1}-\Sigma_{\theta,2}$, which require the positive semi-definite matrices $\partial\Gamma_m/\partial\pi$ and $\partial\Gamma_m/\partial\Delta\Sigma_\theta$, respectively. Likewise as before, the price reversal is larger with noisier supplies and signals in the positive semi-definite matrices of $\partial\Gamma_r/\partial\Sigma_\theta$ and $\partial\Gamma_r/\partial\Sigma_s$.\\
\indent Two known empirical patterns about momentum and reversal are also apparent from the above derivation. The trading volume effect predicts higher autocorrelation for those assets with more trading volume. This pattern is captured by the two partial derivatives $\partial\Gamma_m/\partial X_1>0$ and $\partial\Gamma_r/\partial X_2<0$ where $X_t$ is the vector of asset demands $x_{t,m}$ at time $t$ for asset $m$. A high demand could either be due to a strong signal in high $|s-\bar{D}|$ or a low supply $1+\Theta_t$. In either case, the momentum and the reversal measures are amplified as heterogeneity and stationarity are intensified.\\
\indent Another empirical connection is the price volatility effect, which predicts that the price momentum and reversal will be stronger for assets with higher price volatility. The two partial derivatives $\partial\Gamma_m/\partial\Sigma_{S_1}>0$ and $\partial\Gamma_r/\partial\Sigma_{S_2}<0$ demonstrate this effect where $\Sigma_{S_t}$ is the covariance matrix of the price vector $S_t$ at time $t$. Similar to above, the high volatility of the asset price could be because of a stronger signal or a volatile supply. For both channels, the momentum and reversal patterns are intensified as the heterogeneous beliefs friction and the stationarity increase. The empirical patterns related to the above and the following section are demonstrated by Verardo(2009).

\subsection{Other Price Predictabilities}
Two other price predictability patterns can be easily explained by the multi-asset extension. First, the co-movement in price arises because of correlated payoffs. Suppose now that there are two risky assets with final payoffs of $D_m$ and $D_n$. The two payoffs are correlated with $Cov(D_m,D_n)=\sigma_{m,n}$. Define a co-movement measure $\Gamma_{ct}$ as the covariance matrix of the price vector $(S_t-S_{t-1})$. Hence, if $\Sigma_D$, the covariance matrix of $D$, is positive definite, the co-movement measures $\Gamma_{ct}$ are positive definite for $t=1,2,3$.\\
\indent Intuitively, there are two channels connecting the information between the two risky assets. First, a signal $s_m$ about a risky asset $m$ also acts as a signal for the other risky asset $n$. Second, the price of a risky asset $S_{t,m}$ contains information about the payoff of the other risky asset $n$. Hence, through these two channels, the risky asset prices either move in the same direction if $\sigma_{m,n}>0$ or in the opposite direction if the covariance is negative. The results is a co-movement in prices in each period with the assets that co-move showing price momentum and reversal in the same directions together.\\
\indent The multi-asset extension can also explain the lead-lag effect as well. The effect demonstrates how the price of a risky asset even without any signal is affected by another asset's price. Suppose $\Sigma_D$, the covariance matrix of $D$, is positive definite. If the signals for some but not all the assets are muted such that $s_n=\bar{D}$ for $n=1,\dots,i$ with $i<N$, the co-movement in prices still results. The idea is a simple extension of the previous result on co-movement in prices. Hence, those assets without any good or bad signals may still experience price changes as their valuation follows the leading asset prices.\\
\indent The lead-lag effect may be even stronger than the above proposition. Suppose a signal $s_m$ for a risky asset $m$ is much more precise, that is, more informative than the other risky asset $n$.
\begin{align*}
\frac{\sigma^2_{s^i}}{\sigma^2_{D_i}+\sigma^2_{s^i}}<\frac{\sigma^2_{s^j}}{\sigma^2_{D_j}+\sigma^2_{s^j}}
\end{align*}
Then, depending on the parameters especially the final payoff covariance matrix $\Sigma_D$, it may be the case that even if the asset $n$ had a mildly bad signal, its price may increase if the signal $s_m$ for the asset $m$ is strong and good enough.

\section{Conclusion}
In a simple discrete time framework, the two frictions of asymmetric information and heterogeneous beliefs result in an explanation of price momentum and reversal. Some investors receive a noisy signal about an asset's fundamental value, while others can only infer this signal from the asset price and trade accordingly as trend-chasers. Additionally, the heterogeneous beliefs assumption makes the informed agents under-react because they believe that the uninformed will under-react. Momentum results from the heterogeneous beliefs setup leading to under-reaction of the informed agents, leading to a gradually change in price. Reversal occurs because the price is pushed away from the fundamental value and is pulled back due to the stationarity of the noisy supply and the signal. The model demonstrates how the predictability anomalies of momentum and reversal could arise endogenously in a bounded rationality framework with two frictions. An extension of the model can also account for other stock market patterns such as co-movement and lead-lag effects.

\newpage
\appendix
\section{Single Asset Model}

\begin{lemma} (Verification Theorem)
Under the finite discrete time model with no interim income or consumption, the solutions that maximize the next period returns coincide with those solving the terminal wealth maximization problem.
\end{lemma}
\begin{proof} (Lemma)\\
Let $V_t$ denote the value function for time $t$. In the final period $T$, the value function $V_T$ coincides for both maximization problems. Hence, the solution $x_T^*$ and $S_T^*$ are the same. Suppose the value function solutions for both problems coincided in $t+1$. Then, in $t$, the value function for the terminal wealth maximization is $V_t=E_t(-exp\{-\alpha(W_t+x_t(S_{t+1}^*-S_t)+x_{t+1}^*(S_{t+2}^*-S_{t+1}^*)\})$. Since the equilibrium variables with asterisk are fixed from the previous steps, the expression $W_t+x_{t+1}^*(S_{t+2}^*-S_{t+1})$ in the parenthesis is constant. Hence, the value function first order condition coincides with that from the next period return maximizing value function $V_t=E_t(-exp\{-\alpha(W_t+x_t(S_{t+2}^*-S_t)\})$. By induction, the value function solutions are the same for both problem in all periods.
\end{proof}
\\
\\
\begin{proof} (Proposition 1)\\
For this proof, notice that the subscripts show agent types, and expectations and variances are always conditional on the current period information. The subscripts for variance denote the agent types for taking expectations; for example, $Var_1(D)$ is the variance of $D$ for type 1 agents. Substitution results in for $t=2$
\begin{align*}
\pi\frac{E_1(D)-S_2}{\alpha Var_1(D)}+(1-\pi)\frac{E_2(D)-S_2}{\alpha Var_2(D)}=1+\theta_2
\end{align*}
Again, substituting in the expectations and the variances,
\begin{align*}
&\pi\frac{\bar{D}+\beta_s(s-\bar{D})-S_2}{\alpha Var_1(D)}+(1-\pi)\frac{\bar{D}+\beta_\xi(s-\bar{D}-c_2\theta_2)-S_2}{\alpha Var_2(D)}=1+\theta_2\\
&\pi\frac{\bar{D}+\frac{\beta_s}{b_2}(S_2-a_2)+\beta_sc_2\theta_2-(S_2-a_2)-a_2}{\alpha Var_1(D)}\\
&\qquad\qquad+(1-\pi)\frac{\bar{D}+\frac{\beta_\xi}{b_2}(S_2-a_2)-(S_2-a_2)-a_2}{\alpha Var_2(D)}=1+\theta_2
\end{align*}
The above equation is affine in $((S_2-a_2),\theta_2)$. Collecting the terms, the coefficient for the $(S_2-a)$ term is
\begin{align*}
&\pi\frac{\frac{\beta_s}{b_2}-1}{\alpha Var_1(D)}+(1-\pi)\frac{\frac{\beta_\xi}{b_2}-1}{\alpha Var_2(D)}=0\\
&b_2=\frac{\pi\beta_\xi(\sigma^2_s+c_2^2\sigma^2_\theta)\beta_s+(1-\pi)\beta_s\sigma^2_s\beta_\xi}{\pi\beta_\xi(\sigma^2_s+c_2^2\sigma^2_\theta)+(1-\pi)\beta_s\sigma^2_s}
\end{align*}
For the $\theta_2$ term, the coefficient is
\begin{align*}
&\pi\frac{\beta_sc_2}{\alpha Var_1(D)}=1\\
&c_2=\frac{\alpha\sigma^2_s}{\pi}
\end{align*}
Last, the coefficient for the constant term is
\begin{align*}
&\pi\frac{\bar{D}-a_2}{\alpha Var_1(D)}+(1-\pi)\frac{\bar{D}-a_2}{\alpha Var_2(D)}=1\\
&\pi(1+\frac{\bar{D}-a_2-\alpha Var_1(D)}{\alpha Var_1(D)})+(1-\pi)\frac{\bar{D}-a_2}{\alpha Var_2(D)}=1\\
&\pi-(1-\pi)\frac{\frac{\beta_\xi}{b_2}-1}{\frac{\beta_s}{b_2}-1}\frac{\bar{D}-a_2-\alpha Var_1(D)}{\alpha Var_2(D)}+(1-\pi)\frac{\bar{D}-a_2}{\alpha Var_2(D)}=1\\
&\bar{D}-a_2=\alpha\frac{(\beta_s-b_2)Var_2(D)+(\beta_\xi-b_2)Var_1(D)}{\beta_s-\beta_\xi}\\
&a_2=\bar{D}-\alpha\frac{(\beta_s-b_2)\beta_\xi(\sigma^2_s+c_2^2\sigma^2_\theta)+(\beta_\xi-b_2)\beta_s\sigma^2_s}{\beta_s-\beta_\xi}
\end{align*}

Similarly, for the $t=1$ price,
\begin{align*}
&\pi\frac{E_1(S_2)-S_1}{\alpha Var_1(S_2)}+(1-\pi)\frac{E_2(S_2)-S_1}{\alpha Var_2(S_2)}=1+\theta_1
\end{align*}
Again, substituting in the necessary variables,
\begin{align*}
&\pi\frac{a_2^*+\frac{b_2^*}{b_1}(S_1-a_1)+b_2^*c_1\theta_1-(S_1-a_1)-a_1}{\alpha Var_1(S_2)}\\
&\qquad\qquad +(1-\pi)\frac{a_2+\frac{b_2}{b_1}(S_1-a_1)-(S_1-a_1)-a_1}{\alpha Var_2(S_2)}=1+\theta_1
\end{align*}
The expression for the variance term is
\begin{align*}
Var_1(S_2)=&(b_2^*)^2c_2^2\sigma_{\theta,1}^2\\
Var_2(S_2)=&b_2^2c_2^2\sigma^2_\theta
\end{align*}
The equation is again affine in $((S_1-a_1),\theta_1)$, so collecting the terms for $(S_1-a_1)$,
\begin{align*}
&\pi\frac{b_2^*-b_1}{Var_1(S_2)}+(1-\pi)\frac{b_2-b_1}{Var_2(S_2)}=0\\
&b_1=\frac{\pi b_2^2\sigma^2_\theta b_2^*+(1-\pi)(b_2^*)^2\sigma_{\theta,1}^2 b_2}{\pi b_2^2\sigma^2_\theta+(1-\pi)(b_2^*)^2\sigma_{\theta,1}^2}
\end{align*}
For the $\theta_1$ term, the coefficient is
\begin{align*}
&\pi\frac{b_2^*c_1}{\alpha Var_1(S_2)}=1\\
&c_1=\frac{\alpha b_2^*c_2^2\sigma_{\theta,1}^2}{\pi}
\end{align*}
Finally, the constant term coefficient results in
\begin{align*}
&\pi\frac{a_2^*-a_1}{\alpha Var_1(S_2)}+(1-\pi)\frac{a_2-a_1}{\alpha Var_2(S_2)}=1\\
&a_1=\frac{(b_2-b_1)(a_2^*-\alpha (b_2^*)^2c_2^2\sigma_{\theta,1}^2)+(b_2^*-b_1)(a_2-\alpha b_2^2c_2^2\sigma^2_\theta)}{(b_2-b_1)+(b_2^*-b_1)}
\end{align*}
Last, for $t=0$ combining the market clearing condition and the demand equation,
\begin{align*}
&\pi\frac{a_1^*-S_0}{\alpha Var_1(S_1)}+(1-\pi)\frac{a_1-S_0}{\alpha Var_2(S_1)}=1\\
&S_0=\frac{\pi Var_2(S_1)a_1^*+(1-\pi)Var_1(S_1)a_1-\alpha Var_1(S_1)Var_2(S_1)}{\pi Var_2(S_1)+(1-\pi)Var_1(S_1)}
\end{align*}
\end{proof}
\\
\\
\begin{proof} (Proposition 2)
The covariance term is
\begin{align*}
Cov(S_2-S_1,S_1-S_0)=Cov(S_2-S_1,S_1)=(b_2-b_1)b_1\sigma^2_s-b_1^2c_1^2\sigma^2_\theta
\end{align*}
\begin{align*}
&(b_2-b_1)b_1\sigma^2_s-b_1^2c_1^2\sigma^2_\theta>0\\
&\pi b_2^2\sigma^2_s(b_2-b_2^*)-b_1c_1^2>0\\
&b_2-b_2^*=\frac{\pi\beta_\xi(\sigma^2_s+c_2^2\sigma^2_\theta)\beta_s+(1-\pi)\beta_s\sigma^2_s\beta_\xi}{\pi\beta_\xi(\sigma^2_s+c_2^2\sigma^2_\theta)+(1-\pi)\beta_s\sigma^2_s}\\
&\qquad-\frac{\pi\beta_\xi^*(\sigma^2_s+c_2^2\sigma_{\theta,1}^2)\beta_s+(1-\pi)\beta_s\sigma^2_s\beta_\xi^*}{\pi\beta_\xi^*(\sigma^2_s+c_2^2\sigma_{\theta,1}^2)+(1-\pi)\beta_s\sigma^2_s}\\
&=\pi(1-\pi)\sigma^2_s\beta_s^2(\beta_\xi(\sigma^2_s+c_2^2\sigma^2_\theta)-\beta_\xi^*(\sigma^2_s+c_2^2\sigma_{\theta,1}^2))\\
&\qquad+\pi(1-\pi)\beta_\xi\beta_\xi^*\sigma^2_s\beta_sc_2^2(\sigma_{\theta,1}^2-\sigma^2_\theta)+(1-\pi)^2\beta_s\sigma^4_s(\beta_\xi-\beta_\xi^*)
\end{align*}
$\frac{\partial b_2}{\partial \sigma^2_\theta}<0$. Hence, $b_2^*<b_2$ if $\sigma_{\theta,1}>\sigma^2_\theta$ as is assumed. Further derivation of the condition is not possible as there are many parameters that dictate the condition.
\end{proof}
\\
\\
\begin{proof} (Proposition 3)
The covariance measure for reversal is
\begin{align*}
Cov(D-S_2,S_2-S_1)=-Cov(S_2,S_2-S_1)=-b_2(b_2-b_1)\sigma^2_s-b_2^2c_2^2\sigma^2_\theta
\end{align*}
The measure is always negative as $b_2-b_1>0$ as shown in the previous theorem proof.
\end{proof}

\section{Multiple Asset Model}
\begin{proof} (Proposition 4)\\
The proof is a simple extension of the previous ones to vector notation. Substituting the expectation and variance into the market-clearing condition in $t=2$ gives
\begin{align*}
&\pi(\alpha Var_1(D))^{-1}(E_1(D)-S_2)+(1-\pi)(\alpha Var_2(D))^{-1}(E_2(D)-S_2)=\textbf{1}+\Theta_2\\
&\pi(\alpha Var_1(D))^{-1}(\bar{D}+\beta_s(s-\bar{D})-S_2)\\
&\qquad+(1-\pi)(\alpha Var_2(D))^{-1}(\bar{D}+\beta_\xi(s-\bar{D}-C_2\Theta_2)-S_2)=\textbf{1}+\Theta_2
\end{align*}
Some algebra gives the second period price as
\begin{align*}
S_2=&A_2+B_2(s-\bar{D}+C_2\Theta_2)
\end{align*}
where
\begin{align*}
A_2=&\bar{D}-\alpha(\beta_s-\beta_\xi)^{-1}((\beta_s-B_2)(\Sigma_s+C_2\Sigma_\theta C_2')\beta_\xi'+(\beta_\xi-B_2)\Sigma_s\beta_s')\\
B_2=&(\pi\beta_\xi(\Sigma_s+C_2\Sigma_\theta C_2')+(1-\pi)\beta_s\Sigma_s)^{-1}\\
&\times(\pi\beta_\xi(\Sigma_s+C_2\Sigma_\theta C_2')\beta_s+(1-\pi)\beta_s\Sigma_s\beta_\xi)\\
C_2=&\alpha\pi^{-1}\Sigma_s
\end{align*}
found using the same method as in the single asset case.\\
Similarly, period 1 price comes from the market-clearing condition
\begin{align*}
&\pi(\alpha Var_1(S_2))^{-1}(E_1(S_2)-S_1)\\
&\qquad+(1-\pi)(\alpha Var_2(S_2))^{-1}(E_2(S_2)-S_1)=\textbf{1}+\Theta_1\\
&\pi(\alpha Var_1(S_2))^{-1}(A_2^*+B_1^{-1}B_2^*(S_1-A_1)+B_2^*C_1\Theta_1-(S_1-A_1)-A_1)\\
&\qquad+(1-\pi)(\alpha Var_2(S_2))^{-1}(A_2+B_1^{-1}B_2(S_1-A_1)-(S_1-A_1)-A_1)\\
&\qquad=\textbf{1}+\Theta_1
\end{align*}
Again, some algebraic manipulation gives the price in the first period as
\begin{align*}
S_1=&A_1+B_1(s-\bar{D}+C_1\Theta_1)
\end{align*}
where
\begin{align*}
A_1=&((B_2-B_1)(A_2^*-\alpha B_2^*C_2\Sigma_{\theta,1}(B_2^*C_2)')\\
&+(B_2^*-B_1)(A_2-\alpha B_2C_2\Sigma_\theta(B_2C_2)')((B_2-B_1)+(B_2^*-B_1))^{-1}\\
B_1=&(\pi B_2\Sigma_\theta B_2'+(1-\pi)B_2^*\Sigma_{\theta,1}(B_2^*)')^{-1}\\
&\times(\pi B_2^*B_2\Sigma_\theta B_2'+(1-\pi)B_2B_2^*\Sigma_{\theta,1}(B_2^*)')\\
C_1=&\alpha\pi^{-1}B_2^*C_2C_2^*\Sigma_{\theta,1}
\end{align*}
using the method of undetermined coefficients as before.
Last, the ex-ante price follows from the market-clearing condition
\begin{align*}
\pi(\alpha Var_1(S_1))^{-1}(E_1(S_1)-S_0)+(1-\pi)(\alpha Var_2(S_1))^{-1}(E_2(S_1)-S_0)=\textbf{1}
\end{align*}
The solution gives the price as
\begin{align*}
S_0=&(\pi Var_2(S_1)+(1-\pi)Var_1(S_1))^{-1}(\pi Var_2(S_1)A_1^*\\
&+(1-\pi)Var_1(S_1)A_1-\alpha Var_1(S_1)Var_2(S_1))
\end{align*}
\end{proof}
\\
\\
\begin{proof} (Proposition 5)\\
Using the same steps as in the single asset case, the price autocorrelation $\Gamma_m$ is calculated as follows.
\begin{align*}
Cov(S_2-S_1,S_1-S_0)=Cov(S_2-S_1,S_1)=(B_2-B_1)\Sigma_sB_1'-B_1C_1\Sigma_\theta C_1'B_1'
\end{align*}
The matrix $(B_2-B_1)$ is positive definite, extending the same logic for $b_2-b_1>0$ in the single asset case. Then, the condition for the covariance matrix to be positive semi-definite coincides with that of the single asset case. Again, with many parameters to pin down, the condition for momentum is not precise, but the proposition shows how it is possible to have cross-sectional momentum in the multi-asset case.\\
The reversal measure $\Gamma_r$ is calculated as follows.
\begin{align*}
&Cov(D-S_2,S_2-S_1)=-Cov(S_2,S_2-S_1)\\
&=-(B_2-B_1)\Sigma_sB_2'-B_2C_2\Sigma_\theta(B_2C_2)'
\end{align*}
The measure is always negative semi-definite as $B_1$, $B_2$, $C_2$, $\Sigma_s$, and $\Sigma_\theta$ are all positive semi-definite with the negative sign before each term.
\end{proof}


\newpage
\begin{thebibliography}{99}

\bibitem {AEP1} Abel, A., J. Eberly, and S. Panageas, \textit{Optimal Inattention to the Stock Market}, American Economic Review, 2007.

\bibitem {AEP2} Abel, A., J. Eberly, and S. Panageas, \textit{Optimal Inattention to the Stock Market with Information Costs and Transactions Costs}, Econometrica, 2013.

\bibitem {AM1} Albuquerque, R., and J. Miao, \textit{Advanced Information and Asset Prices}, Journal of Economic Theory, 2014.

\bibitem {AB1} Ang, A. and G. Bekaert \textit{Stock Return Predictability: Is it There?}, Review of Financial Studies, 2007.

\bibitem {AMP1} Asness, C., T. Moskowitz, and L. Pedersen, \textit{Value and Momentum Everywhere}, Journal of Finance, 2013.

\bibitem {BV1} Bacchetta, P. and E. van Wincoop, \textit{Infrequent Portfolio Decisions: A Solution to the Forward Discount Puzzle}, American Economic Review, 2010.

\bibitem {BSV1} Barberis, N., A. Shleifer, and R. Vishny, \textit{A Model of Investor Sentiment}, Journal of Financial Economics, 1998.

\bibitem {BSW1} Barberis, N., A. Shleifer, and J. Wurgler, \textit{Comovement}, Journal of Financial Economics, 2005.

\bibitem {B1} Basak, S., \textit{Asset Pricing with Heterogeneous Beliefs}, Journal of Banking and Finance, 2005.

\bibitem {BGN1} Berk, J., R. Green, and V. Naik, \textit{Optimal Investment, Growth Options, and Security Returns}, Journal of Finance, 1999.

\bibitem {CV1} Cespa, G. and X. Vives, \textit{Dynamic Trading and Asset Prices: Keynes vs. Hayek}, Review of Economic Studies, 2012.

\bibitem {C1} Constantinides, G., \textit{Capital Market Equilibrium with Transaction Costs}, Journal of Political Economy, 1986.

\bibitem {DHS1} Daniel, K., D. Hirshleifer, and A. Subrahmanyam, \textit{Investor Psychology and Security Market Under- and Overreactions}, Journal of Finance, 1998.

\bibitem {DN1} Davis, M. and A. Norman, \textit{Portfolio Selection with Transaction Costs}, Mathematics of Operations Research, 1990.

\bibitem {DT1} De Bondt, W. and R. Thaler, \textit{Does the Stock Market Overreact?}, Journal of Finance, 1985.

\bibitem {DZ1} Du, S. and H. Zhu, \textit{Welfare and Optimal Trading Frequency in Dynamic Double Auctions}, working paper, 2014.

\bibitem {DS1} Duffie, D. and T. Sun \textit{Transactions Costs and Portfolio Choice in a Discrete- Continuous-Time Setting}, Journal of Economic Dynamics and Control, 1990.

\bibitem {FR1} French, K. and R. Roll, \textit{Stock Return Variances: The Arrival of Information and the Reaction of Traders}, Journal of Financial Economics, 1986.

\bibitem {GPR1} Goettler, R., C. Parlour, and U. Rajan, \textit{Equilibrium in a Dynamic Limit Order Market}, Journal of Finance, 2005.

\bibitem {GJM1} Griffin, J., X. Ji, and S. Martin, \textit{Momentum Investing and Business Cycle Risk: Evidence from Pole to Pole}, Journal of Finance, 2003.

\bibitem {G1} Grossman, S., \textit{An Introduction to the Theory of Rational Expectations under Asymmetric Information}, Review of Economic Studies, 1981.

\bibitem {HS1} Hong, H. and J. Stein, \textit{A Unified Theory of Underreaction, Momentum Trading, and Overreaction in Asset Markets}, Journal of Finance, 1999.

\bibitem {HS2} Hong, H. and J. Stein. \textit{Disagreement and the Stock Market}, Journal of Economic Perspectives, 2007.

\bibitem {HL1} Huang, L. and H. Liu, \textit{Rational Inattention and Portfolio Selection}, Journal of Finance, 2007.

\bibitem {JT1} Jegadeesh, N. and S. Titman, \textit{Returns to Buying Winners and Selling Losers: Implications for Stock Market Efficiency}, Journal of Finance, 1993.

\bibitem {JT2} Jegadeesh, N. and S. Titman, \textit{Profitability of Momentum Strategies: An Evaluation of Alternative Explanations}, Journal of Finance, 2001.

\bibitem {J1} Johnson, T., \textit{Rational Momentum Effects}, Journal of Finance, 2002.

\bibitem {JS1} Lewellen, J. and J. Shanken, \textit{Learning, Asset-pricing Tests, and Market Efficiency}, Journal of Finance, 2002.

\bibitem {LM1} Lo, A. and A. MacKinlay, \textit{When are Contrarian Profits due to Stock Market Overreaction?}, Review of Financial studies, 1990.

\bibitem {MR1} Makarov, I. and O. Rytchkov, \textit{Forecasting the Forecasts of Others: Implications for Asset Pricing}, Journal of Economic Theory, 2012.

\bibitem {N1} Nagel, S., \textit{Empirical Cross-Sectional Asset Pricing}, Annual Review of Financial Economics, 2013.

\bibitem {P1} Parlour, C., \textit{Price Dynamics in Limit Order Markets}, Review of Financial Studies, 1998.

\bibitem {PR1} Pindyck, R. and J. Rotemberg, \textit{The Comovement of Stock Prices}, Quarterly Journal of Economics, 1993.

\bibitem {S1} Sims, C., \textit{Implications of Rational Inattention}, Journal of Monetary Economics, 2003.

\bibitem {S2} Shiller, R., \textit{Comovements in Stock Prices and Comovements in Dividends}, Journal of Finance, 1988.

\bibitem {VW1} Vayanos, D. and J. Wang, \textit{Liquidity and Asset Returns Under Asymmetric Information and Imperfect Competition}, Review of Financial Studies, 2012.

\bibitem {VW2} Vayanos, D. and P. Woolley, \textit{Institutional Theory of Momentum and Reversal}, Journal of Finance, 2012.

\bibitem {V1} Verardo, M., \textit{Heterogeneous Beliefs and Momentum Profits}, Journal of Financial and Quantitative Analysis, 2009.

\bibitem {W1} Wang, J., \textit{A Model of Intertemporal Asset Prices Under Asymmetric Information}, Review of Economic Studies, 1993.

\end{thebibliography}
\end{document}